\documentclass{amsart}
\usepackage[foot]{amsaddr}
\usepackage[a4paper,left=4cm,right=4cm]{geometry}
\usepackage{amsmath,amssymb,amsthm}
\usepackage{todonotes}
\usepackage{booktabs}
\usepackage{setspace}
\usepackage{enumerate}
\usepackage{enumitem}
\usepackage{bm}  
\usepackage{algorithm}
\usepackage{url}
\usepackage{xcolor}

\usepackage{graphicx}
\usepackage[section]{placeins} 

\setlist[description]{leftmargin=*}

\newtheorem{thm}{Theorem}[section]

\newtheorem{lem}[thm]{Lemma}

\theoremstyle{definition}

\theoremstyle{remark}

\newcommand{\norm}[1]{\left\Vert#1\right\Vert}

\newcommand{\set}[1]{\left\{#1\right\}}

\newcommand{\RR}{\mathbb{R}}
\newcommand{\Lor}{\mathbb{R}^{1,d}}

\renewcommand{\aa}{\bm{a}}
\newcommand{\bb}{\bm{b}}
\newcommand{\ee}{\bm{e}}
\newcommand{\xx}{\bm{x}}
\newcommand{\yy}{\bm{y}}
\newcommand{\zz}{\bm{z}}
\newcommand{\qq}{\bm{q}}
\newcommand{\uu}{\bm{u}}
\newcommand{\oo}{\bm{o}}

\newcommand{\cB}{\mathcal{B}}

\newcommand{\cH}{\mathcal{H}}

\newcommand{\Rplus}{\mathbb{R}_{\geqslant 0}}

\newcommand{\stress}{\operatorname{Stress}}
\newcommand{\strain}{\operatorname{Strain}}
\newcommand{\dist}{\mathrm{d}}

\DeclareMathOperator{\arcosh}{arcosh}
\DeclareMathOperator{\diag}{diag}

\begin{document}

\title[Hydra: strain-minimizing hyperbolic embedding]{Hydra: A method for strain-minimizing hyperbolic embedding of network- and distance-based data}
\author{Martin Keller-Ressel}
\author{Stephanie Nargang}
\address{Institute of Mathematical Stochastics, Department of Mathematics, TU Dresden, Germany}
\date{\today}
\thanks{The authors would like to thank Carlo Vittorio Cannistraci and Alessandro Muscoloni for comments and discussions.}
\maketitle

\begin{abstract}
We introduce hydra (hyperbolic distance recovery and approximation), a new method for embedding network- or distance-based data into hyperbolic space. We show mathematically that hydra satisfies a certain optimality guarantee: It minimizes the `hyperbolic strain' between original and embedded data points. Moreover, it recovers points exactly, when they are located on a hyperbolic submanifold of the feature space. Testing on real network data we show that the embedding quality of hydra is competitive with existing hyperbolic embedding methods, but achieved at substantially shorter computation time. An extended method, termed hydra+, outperforms existing methods in both computation time and embedding quality.
\end{abstract}

\section{Introduction}

Embeddings of networks and distance-based data into hyperbolic geometry have received substantial interest in recent years. Such embeddings have been used for visualization \cite{walter2004h}, link prediction \cite{papadopoulos2012popularity, papadopoulos2015network} and community detection \cite{papadopoulos2015network, muscoloni2017machine}. They offer insight into the tradeoff between popularity and similarity effects in network growth \cite{papadopoulos2012popularity} and have interesting implications for routing,  network navigability \cite{kleinberg2007geographic, boguna2009navigability} and efficient computation of shortest network paths \cite{zhao2011fast, chowdhary2017improved}. Moreover, such embedding methods can be seen as alternatives to classic visualization and dimensionality reduction techniques based on Euclidean geometry, such as principal component analysis or multidimensional scaling.
However, the hyperbolic embedding methods as yet proposed in the literature have either been based on specific assumptions about network growth (e.g. \cite{papadopoulos2015network, muscoloni2017machine}), or methods with strong theoretical properties, but requiring costly non-linear numerical optimization procedures (e.g. \texttt{H-MDS} of \cite{walter2004h}, \texttt{Rigel} of \cite{zhao2011fast} and \texttt{HyPy} of \cite{chowdhary2017improved}). Here, we introduce \texttt{hydra} (hyperbolic distance recovery and approximation), a novel method for embedding network or distance-based data into hyperbolic space, which has strong mathematical foundations and does not depend on specific assumptions on network growth or structure. At the same time, the method is computationally efficient and based on reduced matrix Eigendecomposition. We show mathematically, that when presented with mutual distances of data points located on a low-dimensional hyperbolic submanifold of the feature space, \texttt{hydra} will recover these points exactly. For general data, the method satisfies a certain optimality property, similar to the strain-minimizing property of multidimensional scaling. Finally, we introduce \texttt{hydra+}, an extension where the result of \texttt{hydra} is used as inital condition for hyperbolic embedding methods based on optimization, such as \texttt{Rigel/HyPy}, substantially improving their efficiency. When tested on real network data, \texttt{hydra} and its variants typically outperform existing hyperbolic embedding methods. All new methods introduced are available in the package \texttt{hydra} \cite{keller-ressel2019hydra} for the statistical computing environment \texttt{R} \cite{rct2016r}.

\section{Embeddings into Hyperbolic Space}

\subsection{Hyperbolic Space}We summarize the key features of the hyperboloid model of hyperbolic geometry (cf. \cite{ratcliffe2006foundations, cannon1997hyperbolic}) in dimension $d$. This will provide the mathematical framework in which we formulate our embedding method. To start, we define for $\xx, \yy \in \RR^{d+1}$ the indefinite inner product\begin{equation}\label{eq:lorentz}
\xx \circ \yy := x_1 y_1 - \left(x_2 y_2 + \dotsc + x_{d+1} y_{d+1}\right),
\end{equation}
also called \emph{Lorentz product}. The real vector space $\RR^{d+1}$ equipped with this inner product is called \emph{Lorentz space} and denoted by $\Lor$. As nested subsets, it contains the \emph{positive Lorentz space} $\Lor_+ = \set{\xx \in \Lor: x_1 >0}$ and the single-sheet hyperboloid
\begin{equation}\label{eq:Hd}
\cH_d = \set{\xx \in \Lor: \xx \circ \xx = 1, x_1 > 0}.
\end{equation}
The \emph{hyperboloid model} with curvature $-\kappa$, $(\kappa >0)$, consists of $\cH_d$ endowed with the hyperbolic distance
\begin{equation}\label{eq:hyper_dist}
\dist^\kappa_H(\xx,\yy) = \frac{1}{\sqrt{\kappa}}\arcosh\left(\xx \circ \yy\right), \qquad \xx, \yy \in \cH_d.
\end{equation}
The hyperbolic distance $\dist^\kappa_H$ is a distance on $\cH_d$ in the usual mathematical sense; in particular it takes only positive values and satisfies the triangle inequality, cf. \cite[\S3.2]{ratcliffe2006foundations}. In fact, it can be shown that $\cH_d$ endowed with the Riemannian metric tensor
\[ds^2 = \frac{1}{\kappa} \left(d\xx \circ d\xx \right)\]
is a Riemannian manifold and $\dist^\kappa_H(\xx, \yy)$ is the corresponding Riemannian distance.\footnote{That is, $\dist^\kappa_H(\xx,\yy)$ is the length of the shortest path from $\xx$ to $\yy$, where lengths are measured using the length element $ds = \sqrt{ds^2}$.} The sectional curvature of this manifold is constant and equal to $-\kappa$, which explains the role of $\kappa$ as curvature parameter. Just as Euclidean space is the canonical model of geometry with zero curvature, hyperbolic space is the canonical model of geometry with negative curvature.

\subsection{The Poincar\'e ball Model}In addition to the hyperboloid model, we introduce the \emph{Poincar\'e ball model} of hyperbolic geometry, which is more appealing for visualizations of hyperbolic space and hyperbolic embeddings. In $\RR^d$, consider the open unit ball 
\[\cB_d := \set{\zz \in \RR^d: |\zz| < 1},\]
where $|\zz| = \sqrt{z_1^2 + \dotsm + z_d^2}$ is the usual Euclidean norm. The \emph{stereographic projection} $\xi: \Lor \to \RR^d$ defined by 
\begin{equation}\label{eq:stereo}
\xi(\xx) = \left(\frac{x_2}{1 + x_1}, \dotsc, \frac{x_{d+1}}{1 + x_1}\right)
\end{equation}
restricts to a bijective mapping from the hyperboloid $\cH_d$ onto $\cB_d$ (cf. \cite[\S4.2]{ratcliffe2006foundations}). It transfers the hyperbolic distance from $\cH_d$ to $\cB_d$, by setting
\[\dist^\kappa_B(\xi(\zz_1), \xi(\zz_2)) = \dist^\kappa_H(\zz_1,\zz_2), \qquad \zz_1, \zz_2 \in \cB_d.\]
Endowed with this distance, $(\cB_d, \dist^\kappa_B)$ is isometric to $(\cH_d, \dist^\kappa_H)$ and therefore an equivalent model of hyperbolic geometry.\\
It will be convenient to parameterize $\cB_d$ by the radial coordinate $r \in [0,1)$ and the directional coordinate $\uu$ (a unit vector in $\RR^d$), given by
\[r := \sqrt{z_1^2 + \dotsm + z_d^2}, \qquad \uu := \frac{\zz}{r}.\]
An easy calculation shows that the conversion from coordinates in $\cH_d$ is given by
\begin{align}\label{eq:conversion_r}
r &= \xi_r(x_1) := \sqrt{\frac{x_1 -1}{x_1 + 1}} \qquad \text{and}\\
\uu &= \xi_{\uu}(x_2, \dotsc, x_{d+1}) := \left(x_2, \dotsc, x_{d+1}\right) / \sqrt{x_2^2 + \dotsm + x_{d+1}^2}. \label{eq:conversion_u}
\end{align}

In dimension $d = 2$, the Poincar\'e ball becomes the \emph{Poincar\'e disc}, and each of its points can be described by the radius $r$ and the unique angle $\theta \in [0,2\pi)$ such that 
\[z_1 = r \cos \theta, \qquad z_2 = r \sin \theta.\]

\subsection{Embedding of Distances and Networks}
To formulate the embedding problem, let a symmetric matrix $D = [d_{ij}] \in \Rplus^{n \times n}$ with zero diagonal be given, which represents the pairwise dissimilarities between some objects $\oo_1, \dotsc, \oo_n$. The basic premise of hyperbolic embedding is that the matrix $D$ can be approximated by a hyperbolic distance matrix $H = [\dist^\kappa_H(\xx_i, \xx_j)]$, i.e., that we can find points $\xx_1, \dotsc, \xx_n$ in low-dimensional hyperbolic space $\cH_d$, such that 
\begin{equation}\label{eq:embed}
\dist^\kappa_H(\xx_i, \xx_j) \approx d_{ij}.
\end{equation}
The points $\xx_1, \dotsc, \xx_n$ give a low-dimensional representation in hyperbolic space of the configuration of $\oo_1, \dotsc, \oo_n$ induced by their dissimilarities. In Euclidean space, such approximations are well studied and can be calculated e.g. by multidimensional scaling (MDS), see also Section~\ref{app:mds} and \cite{borg2005modern}.\\
An important special case is the \emph{network embedding problem}, where a (unweighted, undirected) graph $G = (V,E)$ is given and $D = [d_{ij}]$ is the graph distance matrix of $G$, i.e., $d_{ij}$ is the length of the shortest path in $G$ from vertex $v_i$ to $v_j$. In particular for graphs with \emph{locally tree-like structure} it can be expected that hyperbolic geometry gives a better representation than Euclidean geometry, see e.g \cite{kleinberg2007geographic}. Instead of the shortest-path distance, other dissimilarity measures based on the structure of $G$ can be used, such as the repulsion-attraction (RA) rule or edge-betweenness-centrality (EBC), cf. \cite{muscoloni2017machine}. 

\subsection{Connection to prior work and innovations}
Most existing methods for hyperbolic embedding can be placed into one of two classes: Stress-based methods or network-specific methods. 
 \begin{itemize}[leftmargin=*]

 \item \textbf{Stress-based methods} aim to solve the embedding problem \eqref{eq:embed} by minimizing the \emph{stress functional}
 \begin{equation}\label{eq:stress}
 \stress(\xx_1, \dotsc, \xx_n)^2 := \sum_{i,j=1}^n \left(d_{ij} - \dist^\kappa_H(\xx_i, \xx_j)\right)^2
 \end{equation}
 over all $\xx_1, \dotsc, \xx_n \in \cH_d$. This minimization problem is a challenging high-dimensional non-convex optimization problem, and methods largely differ in their numerical approach to minimize \eqref{eq:stress}. The \texttt{H-MDS} method proposed in \cite{walter2004h} is a gradient descent scheme for minizing \eqref{eq:stress} based on explicit calculation of the gradient. \cite{chamberlain2017neural} propose a neural-network-based approach to minimizing \eqref{eq:stress}, while \cite{zhao2011fast} and \cite{chowdhary2017improved} develop so-called `landmark-based' minimization algorithms (\texttt{Rigel} and \texttt{HyPy} respectively) based on iterative quasi-Newton minimization. Due to the `landmark' heuristic, these methods are able to deal with large-scale instances of \eqref{eq:stress} and do not require full knowledge of $D$, see \cite{chowdhary2017improved} for details.

 \item \textbf{Model-based methods} focus on the network embedding problem and rely on underlying assumptions on the generating mechanism of the graph $G$, see e.g. \cite{papadopoulos2012popularity} for a model of `hyperbolic network growth'. In \texttt{HyperMap} of \cite{papadopoulos2015network} and the \emph{coalescent embedding} of \cite{muscoloni2017machine}, the radial coordinate $r_i$ of the embedded points in the Poincar\'e ball model is determined directly from the degree of the vertices $v_i$, using the assumption of a power-law relationship. The directional component $\uu_i$ of the embedding is then determined by maximizing likelihood in an underlying probabilistic model (cf. \cite{papadopoulos2015network}) or by applying existing nonlinear dimensionality reduction methods (such as Laplacian Eigenmapping or \texttt{ISOMAP}) to the underlying data (cf. \cite{muscoloni2017machine}). 
 \end{itemize}
Here, our main innovation is to replace the stress functional \eqref{eq:stress} by the \emph{strain functional}
\begin{equation}\label{eq:strain}
 \strain(\xx_1, \dotsc, \xx_n)^2 := \sum_{i,j=1}^n \left(\cosh(\sqrt{\kappa}\,d_{ij}) - \xx_i \circ \xx_j\right)^2,
\end{equation}
which results from \eqref{eq:stress} when all distances are transformed by hyperbolic cosine.
Furthermore, we introduce a highly efficient method for the minimization of hyperbolic strain, called \texttt{hydra} (hyperbolic distance recovery and approximation). Contrary to stress-minimization, \texttt{hydra} is based on matrix Eigendecomposition, similar to principal component analysis or classic multidimensional scaling.\footnote{In fact, the relation between hyperbolic strain- and stress-minimization is similar to the relation between `classic' and `metric' multidimensional scaling in the Euclidean case, cf. \cite {borg2005modern}.} In Theorems \ref{thm:exact} and \ref{thm:optimal} we show that \texttt{hydra} satisfies important theoretical optimality properties, in particular, it returns a guaranteed global minimum of \eqref{eq:strain}.
For instances based on real data, the embedding results of \texttt{hydra} are comparable to those based on pure stress-minimization, even when embedding quality is measured by the stress functional \eqref{eq:stress}; see Section~\ref{sec:real} below.  This shows, that even when minimization of stress is the final goal, the strain functional \eqref{eq:strain} is a valuable and useful proxy for stress, as it can be minimized in a highly efficient way. The best results in terms of stress are obtained when strain- and stress-minimization are combined. This is the basis of the \texttt{hydra+} method, introduced in Section~\ref{sec:practical}, where the embedding result of \texttt{hydra} is used as initial condition for a stress-minimization run.
  
 \section{A new hyperbolic embedding method}
\subsection{The hydra algorithm}
We introduce the \texttt{hydra} algorithm, displayed as Algorithm~\ref{algo:hydra}, which calculates an embedding into the Poincar\'e ball model of hyperbolic space by efficiently solving the strain-minimization problem
\begin{equation}\label{eq:strain_min}
\min_{\xx_i \in \Lor} \sum_{i,j} (\cosh(\sqrt{\kappa}\,d_{ij}) - (\xx_i \circ \xx_j))^2.
\end{equation}
The algorithm proceeds as follows: 
\begin{itemize}
\item In steps A1 and A2 the strain-minimization problem is solved by means of a matrix Eigendecomposition. These steps return a coordinate matrix $X = [x_{ij}]$, whose rows $\xx_1, \dotsc, \xx_n$ are elements of positive Lorentz space $\Lor_+$ and the optimizers of \eqref{eq:strain_min}. The optimality of $\xx_1, \dotsc, \xx_n$ is the subject of Theorem~\ref{thm:optimal} below.
\item In steps B1 and B2, the points $\xx_1, \dotsc, \xx_n$ are projected onto the Poincar\'e ball $\cB_d$ using the stereographic projection \eqref{eq:stereo} and converted to radial/directional coordinates $(\uu_i)$ using \eqref{eq:conversion_r} and \eqref{eq:conversion_u}. No adjustment is necessary for the directional coordinates, which are computed in step B1. 
\item Due to \eqref{eq:conversion_r}, the radial coordinates $(r_i)$ depend only on the first column $(x_{11}, \dotsc, x_{n1})$ of $X$ and can be obtained by applying $\xi_r$ elementwise. But  $\xi_r(x_{1i})$ may be undefined for elements with $x_{1i} \in (0,1)$.\footnote{By Theorem~\ref{thm:optimal}, steps A1 and A2 guarantee that all $x_{1i}$ are positive, but not that they are larger than one.} Therefore, $r_i$ is calculated in step B2 as 
\[r_i = \xi_r\left(\frac{x_{i1}}{x_{\min}}\right),\]
that is, after rescaling the first column of $X$ by dividing by its smallest element $x_{\min}$.
\end{itemize}
 
\begin{algorithm}[!htb]
\caption{\texttt{hydra}(D,d,$\kappa$)}
\label{algo:hydra}
\begin{description}
\item[Input] 
\begin{itemize}
\item A symmetric matrix $D = [d_{ij}] \in \Rplus^{n \times n}$ with zero diagonal
\item Embedding dimension $d \le n-1$
\item Hyperbolic curvature parameter $\kappa > 0$\qquad(actual curvature: $-\kappa$)
\end{itemize} 
\item[Step A1] Set 
\begin{equation}\label{eq:Acosh}
A = [a_{ij}] := [\cosh(\sqrt{\kappa} \, d_{ij})]
\end{equation}
 and compute the Eigendecomposition 
\begin{equation}\label{eq:Eigen}
A = Q \Lambda Q^\top,
\end{equation}
where $\Lambda$ is the diagonal matrix of the Eigenvalues $\lambda_1 \ge \dotsm \ge \lambda_n$ and the columns of $Q$ are the Eigenvectors $\qq_1, \dotsc, \qq_n$.
\item[Step A2] Allocate the $n \times (d+1)$-matrix 
\begin{equation}\label{eq:X_hydra}
X := \left[\sqrt{\lambda_1}\,\qq_1 \quad \sqrt{(-\lambda_{n-d+1})^+}\,\qq_{n-d+1} \quad \dotsm \quad \sqrt{(-\lambda_{n})^+}\,\qq_{n} \right],
\end{equation}
where $x^+$ denotes the positive part $x^+ = \max(x,0)$.
\item[Step B1] (`Directional projection') For $i \in 1, \dots, n$ set
\[\uu_i := \frac{(x_{i2}, \dotsc, x_{i(d+1)})}{\sqrt{x_{i2}^2 + \dotsm + x_{i(d+1)}^2}},\]
with $x_{ij}$ the elements of $X$.
\item[Step B2] (`Radial projection') For $i \in 1, \dotsc, n$ set 
\begin{subequations}\label{eq:radial}
\begin{align*}
x_{\min} &:= \min(1, x_{11}, \dotsc, x_{n1})\\
\intertext{and} 
r_i &:= \sqrt{\frac{x_{i1} - x_{\min}}{x_{i1} + x_{\min}}} 
\end{align*}
\end{subequations}
\item[Return] Matrix $X$ and embedding $(r_i, \uu_i)_{i = 1, \dotsc, n}$ as radial and directional coordinates in the Poincar\'e ball $\cB_d$. 
\end{description}
\end{algorithm}

The key theoretical properties of the \texttt{hydra} algorithm are summarized in the following theorems, whose proofs are given in Appendix~\ref{app}. The first theorem shows that $\texttt{hydra}$ recovers any configuration of points in $d$-dimensional hyperbolic space up to isometry:

\begin{thm}[Exact Recovery] \label{thm:exact}
Let $\aa_1, \dotsc, \aa_n$ be points in hyperbolic $d$-space $\cH_d$, and let $D = [d_{ij}] = [\dist^\kappa_H(\aa_i,\aa_j)]$ be the matrix of their hyperbolic distances with curvature $-\kappa$. Then $\texttt{hydra}(D,d,\kappa)$ recovers the points $\aa_1, \dotsc, \aa_n$ up to isometry. In particular, the rows $\xx_1, \dotsc, \xx_n$ of the matrix $X$ and the points $(r_1, \uu_1), \dotsc, (r_n, \uu_n)$ returned by $\texttt{hydra}(D,d,\kappa)$ satisfy 
\[\dist_B^\kappa\Big((r_i, \uu_i),(r_j, \uu_j)\Big) = \dist_H(\xx_i, \xx_j) = d_{ij}, \qquad i,j = 1, \dotsc, n.\]
\end{thm}

For applications to real data, exact recovery is an atypical situation. However, \texttt{hydra} enjoys an optimality guarantee for strain minimization, expressed in the following theorem:

\begin{thm}[Optimal Approximation] \label{thm:optimal}
The rows $\xx_1, \dotsc, \xx_n$ of the matrix $X$ returned by \texttt{hydra}$(D,d,\kappa)$ are the globally optimal solutions of the strain minimization problem \eqref{eq:strain_min}.
Moreover, the first column of $X$ is strictly positive; equivalently, all $\xx_i$ are elements of positive Lorentz space $\Lor_+$. 
\end{thm}

\subsection{Practical guidelines and extensions}\label{sec:practical}
While the result of \texttt{hydra} satisfies the theoretical optimality guarantees in Theorem~\ref{thm:exact} and \ref{thm:optimal}, it can still be advantageous to adjust the results in order to improve the attractiveness of visualization or the embedding quality in terms of stress \eqref{eq:stress} (as opposed to strain, which is globally minimal). The so-called \textit{equiangular adjustment} has been introduced in \cite{muscoloni2017machine} and can be applied to two-dimensional hyperbolic embeddings. Here we propose a slight modification, which allows to interpolate smoothly between no adjustment and full equiangular adjsutment.
\begin{description}
\item[Equiangular adjustment] Let $\lambda \in [0,1]$ be the adjustment parameter and define $\mathrm{ark}(\theta_i)$ as the \emph{angular rank} of $\xx_i$, i.e. when the embedded points are ordered by increasing angular coordinate $\theta$, then $\mathrm{ark}(\theta_i)$ is defined as the rank (from $1$ to $n$) of $\xx_i$ in this list. The adjusted angular coordinate is then set to
\[\theta'_i := \lambda \theta_i + (1-\lambda) (\mathrm{ark}(\theta_i) -1) \frac{2\pi}{n}, \qquad i=1, \dotsc, n.\]
If $\lambda = 0$, no adjustment takes place. If $\lambda = 1$ then the angles $\theta'_i$ are regularly spaced (`equiangular') and only the ordering given by $\theta_i$ is retained.\footnote{This is the equiangular adjustment as proposed in \cite{muscoloni2017machine}.} Values of $\lambda \in (0,1)$ interpolate between these two extremes. We propose a values of $\lambda = 1/2$, which typically leads to improvements in both visual appeal and stress value of the embedding; see also method \texttt{hydra-equi} in Figure~\ref{fig:comp_method}.
\item[hydra+] If minimization of stress is the ultimate objective and strain is used only as a proxy, the result of \texttt{hydra} can be used as an initial condition for a direct minimization of the stress functional \eqref{eq:stress}. This can be seen as a chaining of \texttt{hydra} and \texttt{HyPy/Rigel} \cite{zhao2011fast, chowdhary2017improved}, where \texttt{hydra} substitutes the random initial condition of \texttt{HyPy/Rigel}. For the minimization of stress, efficient quasi-Newton routines, such as LBFGS \cite{zhu1997algorithm} can be used and supplied with the explicit gradient of stress, given in \cite[Eqs.~(3.1),(3.2)]{chowdhary2017improved}.
\end{description}
In terms of efficiency, the following simple improvement can be made to \texttt{hydra}:
\begin{description}
\item[Reduced Eigendecomposition] The numerically dominating part of \texttt{hydra} is the Eigendecomposition in \eqref{eq:Eigen}. Note, however, that in \eqref{eq:X_hydra} only the single first and the last $d$ Eigenvalues and Eigenvectors of the matrix $A$ are needed. There are efficient numerical routines (see e.g. \cite{lehoucq1998arpack}) to perform such a reduced Eigendecomposition without computing the full Eigendecomposition of $A$. These routines substantially improve efficiency if $n \gg d$ and are used in our implementation of \texttt{hydra}. 
\end{description}
Using the reduced Eigendecomposition, we expect the time complexity of \texttt{hydra} to be $\mathcal{O}(n^\alpha)$ with $\alpha$ slightly above, but close, to $2$, cf.~\cite[Ch.~55]{hogben2006handbook}. For \texttt{hydra+}, the time complexity is harder to estimate, since it is based on iterative minimization of a non-convex objective function. In a single step of LBFGS both the stress functional and its gradient have to be evaluated at a complexity of $\mathcal{O}(n^2)$. Depending on the number of steps to convergence, we thus also expect a complexity of $\mathcal{O}(n^\alpha)$, with $\alpha$ strictly larger than $2$. Empirical estimates of $\alpha$ are given in Section~\ref{sec:real} below.

\subsection{Remarks on strain-minimizing graph embeddings}

In the seminal paper \cite{papadopoulos2012popularity} it has been argued that the inherent negative curvature in hyperbolic geometry resolves the trade-off between the conflicting attractive forces of popularity and similarity in network growth models. For this reason \cite{papadopoulos2012popularity} have proposed to interpret the radial coordinate $r$ in the Poincar\'e disc as dimension of `popularity' and the angular coordinate $\theta$ as dimension of `similarity'.
Interestingly, the strain minimization problem \eqref{eq:strain_min} and its solution by \texttt{hydra} gives additional mathematical support for this interpretation. More precisely, revisiting Algorithm~\ref{algo:hydra} in the graph embedding context, we observe that: 
\begin{itemize}[leftmargin=*]
\item The \textbf{radial coordinates} $r_i$ are determined only from the Perron-Frobenius Eigenvector\footnote{The Perron-Frobenius Eigenvector is the Eigenvector associated to the largest Eigenvalue of a positive matrix (i.e. a matrix consisting only of positive entries) and is itself a positive vector, cf. \cite[Ch.~10]{hogben2006handbook}.} $\qq_1$ of the matrix $A$. This provides a remarkable connection to the \emph{Eigenvector centralities} (corresponding to the popularity dimension) of the nodes $v_i$, which are determined from the Perron-Frobenius Eigenvector of their \emph{adjacency matrix}. 
\item The \textbf{directional coordinates} $\uu_i$ are determined only through the Eigenvectors $\qq_{n-d+1}, \dotsc, \qq_n$ (and corresponding Eigenvalues) at the \emph{low end} of the spectrum of $A$. This provides a remarkable connection to Cheeger's inequality (cf. \cite[Ch.~9]{chung2006complex}), which shows that the low end of the spectrum of the graph \emph{Laplacian matrix} encodes the separability of the graph into sparsely connected `communities' (corresponding to the similarity dimension).
\end{itemize}
We remark that while the connections described above are a first step towards a mathematization of the popularity-similarity paradigm in hyperbolic network geometry, the matrix $A = [\cosh\left(\sqrt{\kappa} d_{ij}\right)]$ is in general neither identical to the adjacency nor to the Laplacian matrix of a given graph, and thus further research into the rigorous mathematical underpinning of these connections is warranted. 

\section{Numerical Results}\label{sec:real}
\subsection{Methods and Data}
In our numerical experiments, we evaluate different variants of \texttt{hydra} and compare them to existing hyperbolic embedding methods, using stress as performance criterion. We focus on small to medium sized networks (see Table~\ref{table:networks}), for which it is still feasible to compute the full distance matrix as input to our methods. Edge weights (when available) were discarded, i.e., all networks were treated as unweighted undirected graphs. This network data was used as input for the following methods:
\begin{description}
\item[hydra] The \texttt{hydra} method (without equi-angular adjustment) as described in Algorithm~\ref{algo:hydra}
\item[hydra-equi] The \texttt{hydra} method with equiangular adjustment $\lambda = 0.5$, as described in Section~\ref{sec:practical}
\item[hydra+] The \texttt{hydra+} method as described in Section~\ref{sec:practical} and using the result of \texttt{hydra-equi} as initial condition.
\item[HyPy/Rigel]  The \texttt{HyPy} algorithm from \cite{chowdhary2017improved}, which is based on \texttt{Rigel} from \cite{zhao2011fast}. Both methods are based on direct minimization of the stress functional \eqref{eq:stress}. \textit{Landmark selection}, as proposed in \cite{chowdhary2017improved} was not implemented, since it serves to reduce runtime and memory use for large networks, but is not expected to improve embedding results. As in \cite{chowdhary2017improved}, the initial condition for minimization was chosen at random and we repeated the embedding $100$ times.
\item[CE-LE] The coalescent embedding (CE) using Laplacian Eigenmapping (LE) as dimension-reduction method, full equiangular adjustment and repulsion-attraction (RA) pre-weighting; see \cite{muscoloni2017machine} for details. Among the methods developed in \cite{muscoloni2017machine}, this was the best performing method to invert the PSO generating model of \cite{papadopoulos2015network} for hyperbolic networks.
\end{description}
For the methods \texttt{hydra}, \texttt{hydra-equi}, \texttt{hydra+} and \texttt{HyPy/Rigel} we used our own implementations in \texttt{R}, which are available in the \texttt{R}-package \texttt{hydra}, \cite{keller-ressel2019hydra}. For \texttt{CE-LE} we used the MATLAB implementation of the methods of \cite{muscoloni2017machine} available from github.\footnote{\url{https://github.com/biomedical-cybernetics/coalescent_embedding}} The stress-optimization in \texttt{hydra+} and \texttt{HyPy/Rigel} was performed using the LBFGS method (see \cite{zhu1997algorithm}) as implemented in the \texttt{R}-function \texttt{optim} and using the analytic form of the gradient of the stress functional \eqref{eq:stress} from \cite{chowdhary2017improved}. Note that all methods except \texttt{CE-LE} use the shortest-path matrix as input dissimilarities; \texttt{CE-LE} uses repulsion-attraction (RA) weights as input dissimilarities, see \cite{muscoloni2017machine}. For all methods hyperbolic curvature was fixed to $-\kappa = -1$ and we embed into dimension $d = 2$. 

\begin{table}[htpb]
\begin{footnotesize}
\begin{tabular}{@{}lp{0.5\textwidth}ll@{}}
\toprule
Network&Description&Source&\#\,Nodes\\
\midrule
karate & Social interaction network (`Zachary's karate club network') from \cite{zachary1977information}&\texttt{igraphdata} \cite{csardi2015igraphdata}&$34$\\
UKfaculty & Personal friendship network of a UK university faculty from \cite{nepusz2008fuzzy} & \texttt{igraphdata} \cite{csardi2015igraphdata} & $81$\\
opsahl & One-node projection of message Exchange Network from \cite{opsahl2013triadic}; two isolated nodes have been removed&\url{toreopsahl.com}&$897$\\
facebook & Facebook social circles network from \cite{leskovec2012learning}; combined edge sets&\url{snap.stanford.edu}&$4039$\\
collaboration & Co-authorship network from ArXiv submissions to category Hep-Ph (High Energy Physics); largest connected component. From \cite{leskovec2007graph}& \url{snap.stanford.edu} &$8638$\\
oregon & Autonomous systems peering information inferred from route-views in Oregon on March 26, 2001. From \cite{leskovec2005graphs}& \url{snap.stanford.edu} &$11174$\\
\bottomrule
\end{tabular}
\vspace{0.5em}
\caption{Networks used for numerical experiments}\label{table:networks}
\end{footnotesize}
\end{table}

\subsection{Results and Discussion}
Results on embedding quality (measured by stress) for all networks and methods (except \texttt{CE-LE}) are shown in Figure~\ref{fig:comp_method}. Note that stress values are normalized by using the average result of \texttt{HyPy/Rigel} as a reference. This facilitates the comparison of results between different networks. As \texttt{HyPy/Rigel} depends on randomized initial conditions we indicate the $5\%$- and $95\%$-quantiles (over 100 runs) in addition to its average result. The stress-values of the embeddings produced by \texttt{CE-LE} were substantially larger (by a factor from $13$ to $26$) than the reference method and we have therefore excluded this method from the plot and from the further analysis of computation times. \\
\begin{figure}[htbp]
{\centering
\includegraphics[width=1.15 \textwidth]{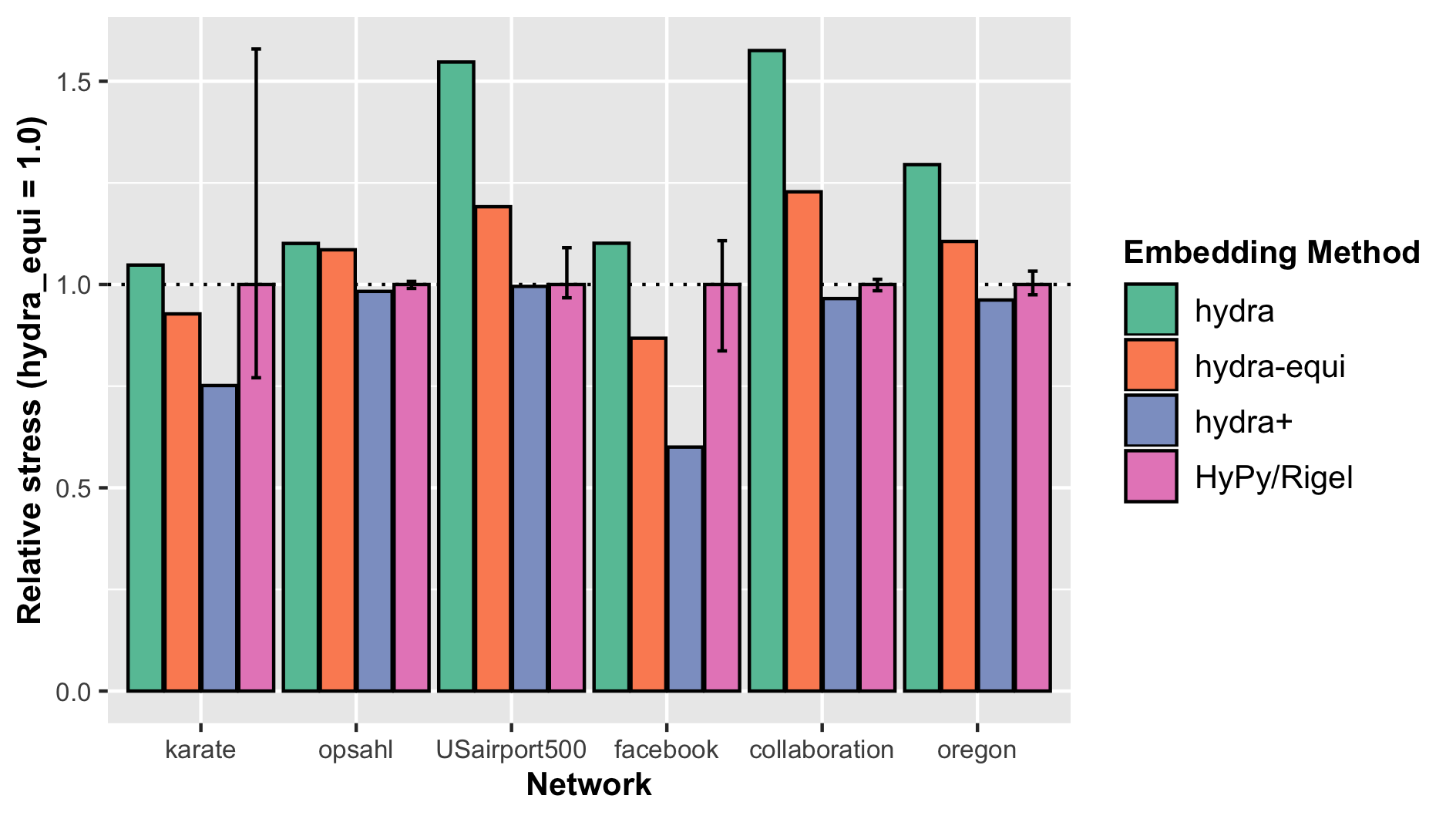}
}
\caption{\textbf{Embedding performance on real network data.} Embedding quality (measured by stress \eqref{eq:stress}, relative to the average result of \texttt{HyPy/Rigel}) of different hyperbolic embedding methods applied to the six networks listed in Table~\ref{table:networks}. For \texttt{HyPy/Rigel} a 5\%--95\% error bar is shown, corresponding to 100 runs with randomized initial condition.}
\label{fig:comp_method}
\end{figure}

Computation times of the different methods is shown in Figure~\ref{fig:comp_time} in doubly logarithmic coordinates. As should be expected, the methods split into two groups with computation times for \texttt{hydra} and \texttt{hydra-equi} being shorter than for \texttt{hydra+} and \texttt{HyPy/Rigel} by two orders of magnitude. The seemingly small gap between \texttt{hydra+} and \texttt{HyPy/Rigel} still corresponds to a difference of about 50\% in runtime. Based on the discussion in Section~\ref{sec:practical} we have added regression lines to estimate the exponent $\alpha$ in the conjectured complexity $\mathcal{O}(n^\alpha)$. To avoid clutter, regression lines are only shown for \texttt{hydra-equi} and \texttt{hydra+}; the estimates for all methods are $\alpha \approx 2.0$ for \texttt{hydra}, $\alpha \approx 2.1$ for \texttt{hydra-equi} and $\alpha \approx 2.3$ for both \texttt{hydra+} and \texttt{HyPy/Rigel}. As setup costs seem to dominate the computation times for the smallest network, we have excluded it from the regression analysis. 
\begin{figure}[htbp]
{\centering
\includegraphics[width=1.15 \textwidth]{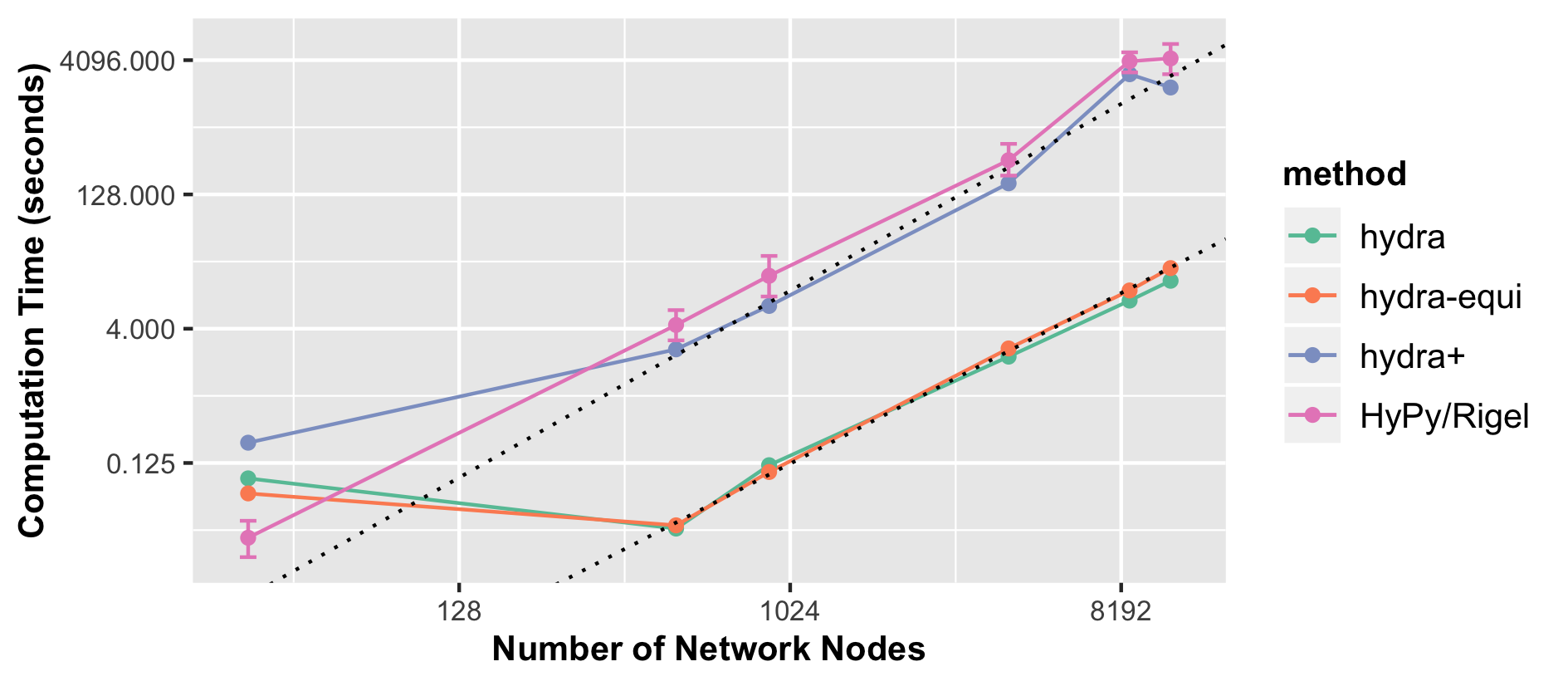}
}
\caption{\textbf{Computation time of embedding methods.} Computation time (in seconds) of different hyperbolic embedding methods, in relation to the number of nodes in the six networks listed in Table~\ref{table:networks}. Coordinate axes are doubly logarithmic. For \texttt{HyPy/Rigel}, average computation time and a 5\%--95\% error bar is shown, corresponding to 100 runs with randomized initial condition. For \texttt{hydra-equi} and \texttt{hydra+} dotted regression lines (excluding observations from the smallest network) are indicated.}
\label{fig:comp_time}
\end{figure}

Finally, an exemplary embedding result produced by the \texttt{hydra-equi} method for the \texttt{facebook} network is shown in Figure~\ref{fig:example}.  Nodes are placed into the Poincar\'e disc model of hyperbolic geometry according to their embedding coordinates (the full disc is indicated in grey) and a random subsample of links is drawn as hyperbolic geodesics. Visually, the embedding conforms well with the popularity-vs-similarity paradigm of \cite{papadopoulos2012popularity} for hyperbolic networks: Nodes with a function as hubs between communities or individuals (popularity dimension) are placed closer to the center of the hyperbolic disc. Communities are identified along the angular coordinate (similarity dimension) with the effective distance between communities indicated by angular separation.\\

Summarizing our numerical experiments, we conclude the following:
\begin{itemize}
\item In general, the strain-minimization performed by \texttt{hydra} seems to be a good proxy for stress-minimization, but is faster by a factor of 100 or more in comparison to stress-minimization from a random initial condition (\texttt{HyPy/Rigel}). Note that \texttt{hydra} also eliminated the uncertainty associated with the randomized nature of \texttt{HyPy/Rigel}, which can lead to large variations in embedding quality in some instances (e.g., the \texttt{facebook} network).
\item The simple equi-angular adjustment performed in \texttt{hydra-equi} consistently improves embedding quality in terms of stress at  negligible numerical costs. The returned embeddings outperform \texttt{HyPy/Rigel} for two networks (\texttt{karate}, \texttt{facebook}) and are competitive for all others, with a largest observed difference of 23\% in terms of stress.
\item Using the results of \texttt{hydra} as a starting value for stress-minimization, instead of a random initial condition, i.e., replacing \texttt{HyPy/Rigel} by \texttt{hydra+} reduces computation time by approx. 30\% - 50\% and leads to better (average) embedding quality in all cases. The reduction in stress is considerable in the facebook network, where stress is reduced by approx. 40\%.
\item The \texttt{CE-LE} method, based on the PSO network growth model of \cite{papadopoulos2015network}, is not competitive with the other methods in terms of embedding quality. This suggests that the structure of the real networks that we have considered deviates from the theoretical growth model (PSO-model) of \cite{papadopoulos2015network} upon which \texttt{CE-LE} is build. 
\end{itemize}

As a next step, we plan to make strain-minimizing hyperbolic embedding methods feasible for large and very large networks. For such networks the computational complexity of $\mathcal{O}(n^\alpha)$ (with $\alpha > 2$) of the proposed methods, but also of the graph distance calculation itself, are prohibitive. For this reason, heuristics such as the landmark heuristic of \cite{chowdhary2017improved} will have to be adapted to strain-minimizing embedding methods.

%

\begin{figure}
{\centering
\includegraphics[width=1.0 \textwidth]{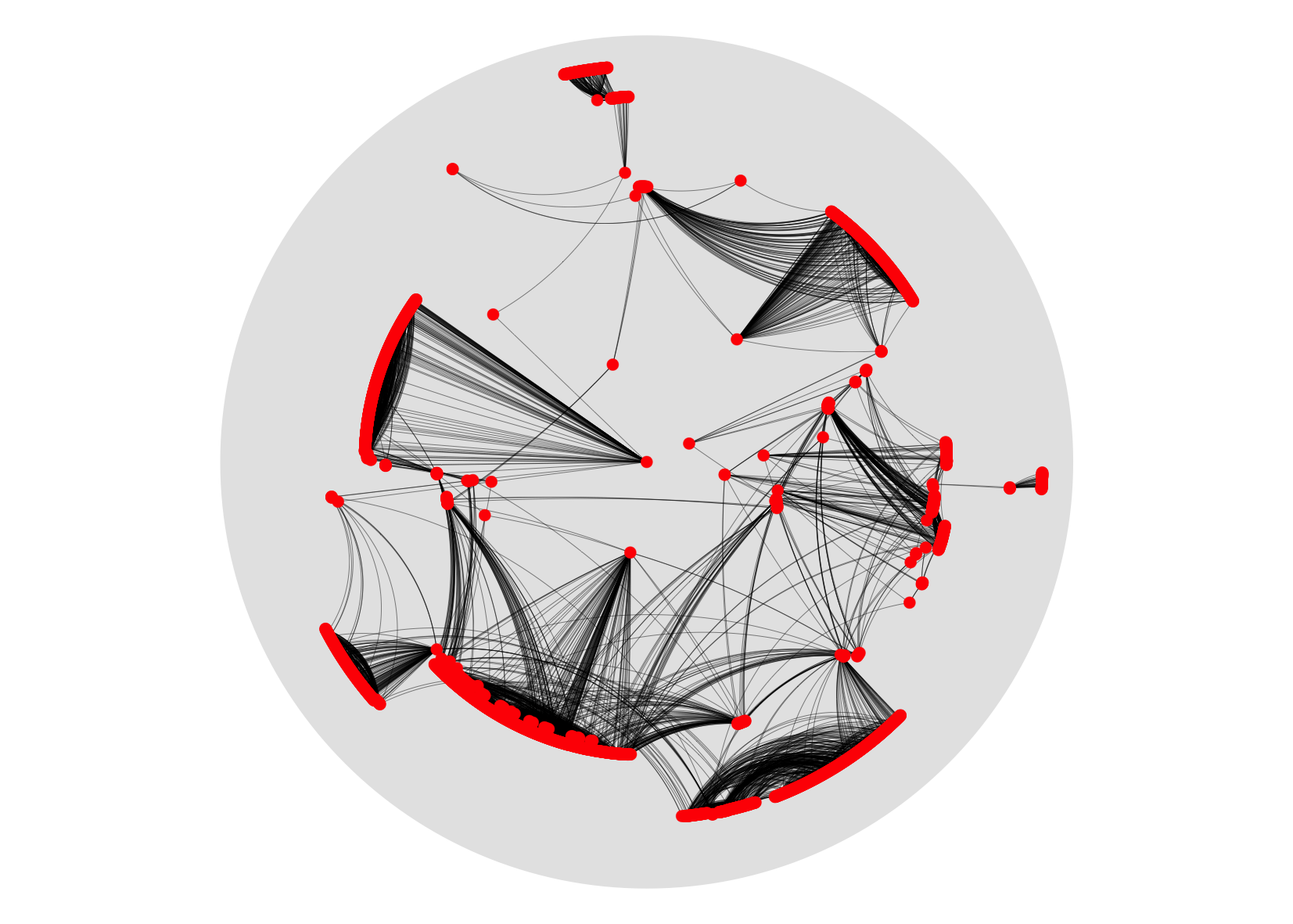}
}
\caption{\label{fig:facebook}\textbf{Embedding example.} The hyperbolic embedding of the \texttt{facebook} network produced by the method \texttt{hydra-equi}. All 4039 network nodes are shown as red dots. A random subsample of the 88234 total edges are also shown and drawn as hyperbolic geodesics in black. The edge subsample was produced by randomly sampling two incident edges from each network node, allowing for repetitions.}
\label{fig:example}\end{figure}

\begin{appendix}

\section{Theoretical Results}\label{app}
To prove the theoretical properties of the \texttt{hydra} method, it is convenient to reformulate the strain minimization problem \eqref{eq:strain_min} in matrix form. To this end, let $D = [d_{ij}]$ be the given dissimilarity matrix, set $A = [\cosh(\sqrt{\kappa}\,d_{ij})]$ and write 
\[X = \left(\xx_1, \dotsc, \xx_n\right)^\top \in \RR^{n \times (d+1)}\]
for the coordinate matrix of some points $\xx_1, \dotsc, \xx_n$ in $\RR^{d+1}$. Finally, let $J$ be the $(d+1) \times (d+1)$ diagonal matrix
\begin{equation}\label{eq:J}
J = \diag(1,-1, \dotsc, -1),
\end{equation}
cf. \cite[\S3.1]{ratcliffe2006foundations}. The strain minimization problem \eqref{eq:strain_min} can now be written in compact form as 
\begin{equation}\label{eq:strain_compact}
\min_{X \in \RR^{n \times (d+1)}}\norm{A - X^\top J X}_F^2,
\end{equation}
where $\norm{.}_F$ denotes the Frobenius norm. Imposing the constraint that all $\xx_i$ are elements of the hyperboloid  $\cH_d$ is equivalent to requiring that
\[\diag(X^\top J X) = (1, \dotsc, 1) \quad \text{and} \quad X\ee_1 > 0,\]
where $\ee_1$ is the first standard unit vector. In particular, the first condition guarantees $\xx_i \circ \xx_i = 1$, and the second one selects the upper sheet of the two-sheet hyperboloid thus described.

\subsection{Hyperbolic strain minimization and exact recovery} 
For a real symmetric matrix $A$, denote by $n_+(A)$ and $n_-(A)$ the number of positive and negative Eigenvalues of $A$. 
The following Lemma characterizes matrices that can be written as inner product matrices (`Gram matrices') with respect to the Lorentz product \eqref{eq:lorentz}:
\begin{lem}\label{lem:gram}
Let $G = [g_{ij}] \in \Rplus^{n \times n}$ be positive and symmetric, and let $d \le n-1$. The following are equivalent
\begin{enumerate}[label=\alph*)]
\item $G$ satisfies $n_+(G) = 1$ and $n_-(G) \le d$. 
\item $G$ is a `Lorentzian Gram matrix', i.e., there exist $\xx_1, \dotsc, \xx_n$ in $\Lor$, such that
\[g_{ij} =  \xx_i \circ \xx_j, \quad \forall\,i,j \in 1, \dotsc, n.\]
\item There exists $X \in \RR^{n \times (d+1)}$, such that 
\begin{equation}\label{eq:G_decomp}
G = X J X^\top,
\end{equation}
where $J$ is given by \eqref{eq:J}.
\end{enumerate}
In addition, 
\begin{itemize}
\item The first column of $X$ is positive if and only if $\xx_1, \dotsc, \xx_n$ are in the positive Lorentz space $\Lor_+$;
\item The points $\xx_1, \dotsc, \xx_n$ are in $\cH_d$ if and only if $\diag(G) = (1, \dotsc, 1)$ and the first column of $X$ is positive.
\end{itemize} 
\end{lem}
\begin{proof}
The equivalence of (b) and (c) follows directly from the definition of the Lorentz product in \eqref{eq:lorentz}. Next, we show that (c) implies (a): From \cite[Ch.~10.3]{lax2007linear} it follows from \eqref{eq:G_decomp} that $n_+(G) \le n_+(J) = 1$ and $n_-(G) \le n_-(J) = d$. But $G$ is a positive matrix and Perron's theorem (cf. \cite[Ch.~16]{lax2007linear}) guarantees that its leading Eigenvalue is positive, i.e., $n_+(G) \ge 1$, and we conclude (a). To show that (a) implies (c), assume first that $n_-(G) = d$. By Sylvester's law of inertia, there exists a decomposition 
\[G = \hat{X} \hat{J} \hat{X}^\top, \quad \text{where} \quad \hat{J} =\diag\Big({+}1, 0, \dotsc, 0,\underbrace{-1, \dotsc, -1}_\text{$d$ times}\Big).\]
This decomposition can be reduced to \eqref{eq:G_decomp}, by simply dropping all rows and columns containing only zeroes from $\hat{J}$ and by also dropping the corresponding columns from $\hat{X}$. If $n_-(G) = d' < d$, the same procedure yields a decomposition with $X$ of dimension $n \times (d' + 1)$ and $J$ of dimension $(d'+1) \times (d'+1)$. Padding $X$ with zero columns and $J$'s diagonal with $-1$s, \eqref{eq:G_decomp} also follows in this case.\\
The additional statements follow directly from the following observations: The first column of $X$ contains exactly the first coordinate of all points $\xx_1, \dotsc, \xx_n$. If the first coordinate of a point $\xx$ is positive, it is an element of positive Lorentz space and vice versa. The diagonal of $G$ contains the values $\xx_i \circ \xx_i, i = 1, \dotsc, n$. If $\xx_i \circ \xx_i = 1$ and $\xx_i \in \Lor_+$ then $\xx_i$ is an element of the hyperboloid $\cH_d$ and vice versa.  
\end{proof}

\begin{proof}[Proof of Theorem~\ref{thm:optimal}]Let $A = [a_{ij}] = [\cosh(\sqrt{\kappa}\,d_{ij})]$ and let $B = [b_{ij}]$ be another symmetric matrix in $\RR^{n \times n}$. Let $(\lambda_i(A))_{i = 1, \dotsc, n}$ and $(\lambda_i(B))_{i = 1, \dotsc, n}$ be their  Eigenvalues in descending order, and denote by $\norm{.}_F$ the Frobenius norm. By a result of Wielandt-Hoffmann, cf. \cite[Ch.~10, Thm.~18]{lax2007linear},
\begin{equation}\label{eq:wielandt}
\sum_{i,j} (a_{ij} - b_{ij})^2 = \norm{A - B}_F^2 \ge \sum_{i}^n (\lambda_i(A) - \lambda_i(B))^2.
\end{equation}
Assume now that $B$ is a `Lorentzian Gram matrix'  with elements given by 
\[b_{ij} = \bb_i \circ \bb_j, \qquad i,j=1, \dotsc, n\]
for some $\bb_1, \dotsc, \bb_n  \in \Lor$. By Lemma~\ref{lem:gram} this implies that $n_+(B) = 1$ and $n_-(B) \le d$. Hence all Eigenvalues of $B$ with index $2, \dotsc, n-d$ are zero, and we obtain
\begin{align*}
\sum_{i,j} (a_{ij} - \bb_i \circ \bb_j)^2 &= \norm{A - B}_F^2  \ge \notag \\
&\ge (\lambda_1(A) - \lambda_1(B))^2 + \sum_{i=2}^{n-d} \lambda_i(A)^2 + \sum_{i = n-d+1}^n (\lambda_i(A) - \lambda_i(B))^2.
\end{align*}
For the first summand on the right hand side we have the trivial lower bound $0$. In the last sum, all $\lambda_i(B)$ are negative or zero, and hence, for any $i = (n-d+1), \dotsc, n$, we can estimate
\[(\lambda_i(A) - \lambda_i(B))^2 \ge \begin{cases} 0 \quad &\text{if } \lambda_i(A) \le  0\\ \lambda_i(A)^2 \quad &\text{if }\lambda_i(A) > 0, \end{cases}\]
which is the same as $(\lambda_i(A)^+)^2$.
Together, we obtain that
\begin{equation}\label{eq:strain_bound}
\sum_{i,j} (a_{ij} - \bb_i \circ \bb_j)^2 \ge \sum_{i=2}^{n-d} \lambda_i(A)^2 + \sum_{i=n-d+1}^{n} (\lambda_i(A)^+)^2.
\end{equation}
Denote by $A = Q \Lambda_A Q^\top$ the Eigendecomposition of $A$ with $\Lambda_A = \diag(\lambda_1(A), \dotsc, \lambda_n(A))$. Let $X$ be the matrix returned by \texttt{hydra}($D$, $d$, $\kappa$) and $\xx_1, \dotsc, \xx_n$ the rows of $X$. By \eqref{eq:X_hydra} the associated Lorentzian Gram matrix $G = XJ X^\top$ has the Eigendecomposition $G = Q \Lambda_G Q^\top$ with
\[\Lambda_G = \diag(\lambda_1(A), 0, \dotsc, 0, (-\lambda_{n-d+1}(A))^+, \dotsc, (-\lambda_n(A))^+).\]
Using the unitary invariance of the Frobenius norm and the trivial identity $x - (-x)^+ = x^+$, we obtain
\begin{align}\label{eq:x_minimizes}
\sum_{i,j} (a_{ij} - \xx_i \circ \xx_j)^2 &= \norm{Q \Lambda_A Q^\top - Q \Lambda_G Q^\top}^2_F = \norm{\Lambda_A - \Lambda_G}_F^2 = \\ = \sum_{i=2}^{n-d} \lambda_i(A)^2 + \sum_{i=n-d+1}^n (\lambda_i(A)^+)^2.\notag
\end{align}
This shows that setting $\bb_i  := \xx_i$ for all $i \in 1, \dotsc, n$ achieves equality in \eqref{eq:strain_bound} and hence that the points $\xx_i$ minimize \eqref{eq:strain_compact}.
\end{proof}

\begin{proof}[Proof of Theorem~\ref{thm:exact}]
Let $D = [d_{ij}]$ be the hyperbolic distance matrix of $\aa_1, \dotsc, \aa_n$ in $\cH^d$. Then $A = [a_{ij}] = [\cosh(\sqrt{\kappa}d_{ij})]$ is the associated Lorentzian Gram matrix with elements
\[a_{ij} = \aa_i \circ \aa_j.\]
By Lemma~\ref{lem:gram} $A$ satisfies $n_+(A) = 1$ and $n_-(A) \le d$, i.e. the Eigenvalues of $A$ satsify $\lambda_i(A) = 0$ for $i=2, \dotsc, n-d$ and $\lambda_i(A) \le 0$ for $i=n-d+1, \dotsc, n$. Hence, it follows from \eqref{eq:x_minimizes} that $\sum_{i,j}(a_{ij} - \xx_i \circ \xx_j)^2 = 0$ or,  equivalently, that
\[\xx_i \circ \xx_j = \aa_i \circ \aa_j\]
for all $i, j \in 1, \dotsc, n$. Applying $\cosh(\sqrt{\kappa}\,\cdot)$ to both sides, we see that
\[\dist_H^\kappa(\xx_i, \xx_j) = \dist_H^\kappa(\aa_i ,\aa_j)\]
and hence that $(\xx_i)$ and $(\aa_i)$ are isometric. 
\end{proof}

\subsection{Comparison to classic multidimensional scaling}\label{app:mds}
In several aspects, the hydra method can be seen as the `hyperbolic analogue' of classic multidimensional scaling (MDS), cf. \cite{borg2005modern}, which is based on Euclidean geometry. Below, we summarize the classical MDS method and point out parallels to (and differences from) \texttt{hydra}. Classical MDS also takes a matrix $D = [\dist_{ij}] \in \Rplus^{n \times n}$ with zero diagonal as input. Using the centering matrix $C = I - \frac{1}{n}\bm{1} \in \RR^{n \times n}$,  where $\bm{1}$ denotes a matrix of ones of matching dimension,  the `doubly centered' matrix 
\begin{equation*}
A = - \frac{1}{2}C^\top D C \qquad \text{[compare \eqref{eq:Acosh}]}
\end{equation*}
is derived from $D$, and its Eigendecomposition
\[A = Q \Lambda Q^\top \qquad \text{[compare \eqref{eq:Eigen}]}\]
computed. Again, $\Lambda$ is the diagonal matrix of the Eigenvalues $\lambda_1 \ge \dotsm \ge \lambda_n$ and the columns of $Q$ are the Eigenvectors $\qq_1, \dotsc, \qq_n$. MDS then returns the (Euclidean) coordinate matrix
\[X = \left[\sqrt{\lambda_1}\,\qq_1 \quad \sqrt{\lambda_2}\,\qq_2 \quad \dotsm \quad \sqrt{\lambda_d}\,\qq_d \right],\qquad \text{[compare \eqref{eq:X_hydra}]}\]
whose rows $\xx_i$ are interpreted as points in Euclidean space $\RR^d$. This coordinate matrix $X$ solves the strain minimization problem
\[\min_{X \in \RR^{n \times d}} \norm{A  - X^\top X}_F^2, \qquad \text{[compare \eqref{eq:strain_compact}]}\]
cf. \cite[Ch.~12]{borg2005modern}. Moreover, if the input matrix $D$ is a matrix of \emph{squared} Euclidean distances, i.e., $\dist_{ij} = |\xx_i - \xx_j|^2$ then MDS recovers the points  $\xx_i$ exactly (up to Euclidean isometry). Note that $X^\top X$ appearing above is the Gram matrix of the points $\xx_1, \dotsc, \xx_n$, i.e. the matrix of their scalar products $\xx_i^\top \xx_j$, whereas the matrix $X^\top J X$ in \eqref{eq:strain_compact} is the `Lorentzian Gram matrix' of the Lorentz products $\xx_i \circ \xx_j$.

\bibliographystyle{plain}
\bibliography{references}
\end{appendix}
\end{document}